\documentclass[12pt]{amsart}
\usepackage{amsmath,amsfonts,amssymb,amsthm,a4wide}
\usepackage{pst-node}

\DeclareMathOperator{\val}{val}
\DeclareMathOperator{\fac}{Fac}
\theoremstyle{plain}
\newtheorem{theorem}{Theorem}
\newtheorem{lemma}[theorem]{Lemma}

\newtheorem{prop}[theorem]{Proposition}
\newtheorem{properties}[theorem]{Properties}
\theoremstyle{definition}
\newtheorem{definition}[theorem]{Definition}
\newtheorem{question}[theorem]{Question}
\newtheorem{remark}[theorem]{Remark}

\newtheorem{example}[theorem]{Example}

\title{Relations on words}
\author{Michel Rigo}
\address{University of Liege, Dept. of Math., All\'ee de la d\'ecouverte 12 (B37), B-4000 Li\`ege, Belgium.}
\email{M.Rigo@ulg.ac.be}
\thanks{dedicated to the memory of my grandmother Suzanne Wiame 1932--2016.}
\date{\today}

\begin{document}

\maketitle

\begin{abstract}
    In the first part of this survey, we present classical notions arising in combinatorics on words: growth function of a language, complexity function of an infinite word, pattern avoidance, periodicity and uniform recurrence. Our presentation tries to set up a unified framework with respect to a given binary relation. 

In the second part, we mainly focus on abelian equivalence, $k$-abelian equivalence, combinatorial coefficients and associated relations, Parikh matrices and $M$-equivalence. In particular, some new refinements of abelian equivalence are introduced.
\end{abstract}

\section{Introduction}
This paper follows and complements the talk I gave during the conference on {\em Automatic Sequences, Number Theory, and Aperiodic Order} held in Delft in October 2015. The aim is to survey various concepts arising in combinatorics on words and present them in a unified and general framework. In Section~\ref{sec:gen}, {\em relatively to a given binary relation over $A^*$}, we define the growth function of a language, the complexity function of an infinite word and the notions of avoidable patterns, periodicity and uniform recurrence. These notions are usually first introduced for the restrictive case of equality of factors, e.g., the complexity function counts the number of factors of length $n$ occurring in a given infinite word but we could count them up to rearrangement of the letters. 
In the second part of the paper, we review classical binary relations on words where these concepts may be applied. In Section~\ref{sec:a1}, we consider abelian equivalence, then its extension to $k$-abelian equivalence is presented in Section~\ref{sec:a2}. We pursue in Section~\ref{sec:a4} with binomial coefficients of words and various equivalence relations that can be associated with. In Section~\ref{sec:a3}, we present Parikh matrices and related relations. In the last section, we briefly present partial words and their generalizations to similarity relations.   
The bibliography is not exhaustive (it is limited to 100 entries) but we hope that it could provide relevant entry points to the existing literature. We limit ourselves to the unidimensional case. Indeed, many of the presented concepts have counterparts in a multidimensional setting.

\section{Basics}
We give some basic definitions about words. For general references, see \cite{CANT,Lot,Rigo}. Let $A$ be a finite alphabet, i.e., a finite set of elements called letters. A {\em finite word} $w$ over $A$ is a finite sequence of elements in $A$. So it is a map $w:\{1,\ldots,n\}\to A$ where $n\in\mathbb{N}$ is the length of the word $w$. In particular, the empty sequence is called the {\em empty word} and is denoted by $\varepsilon$. Its length is $0$. Note that the indexing of finite words begins at position $1$. The set of finite words over $A$ is denoted by $A^*$. Endowed with the concatenation of words as product operation, $A^*$ is a monoid with $\varepsilon$ as neutral element. We write $|w|$ for the length of the word $w$ and $|w|_a$ for the number of occurrences of the letter $a$ in $w$. We directly have $|w|=\sum_{a\in A}|w|_a$. 
Let $n$ be an integer. We write $A^n$ to denote the set of words of length $n$ over $A$. Let $A,B$ be two finite alphabets. A map $f:A^*\to B^*$ is a {\em morphism} (of monoids) if $f(uv)=f(u)f(v)$ for all $u,v\in A^*$ and, in particular, we have $f(\varepsilon)=\varepsilon$. A morphism $f$ is {\em non-erasing} if $f(u)\neq\varepsilon$ for all non-empty words $u\in A^*$. A morphism is characterized by the images of the letters of its domain. If the images of the letters all have length $1$, the morphism is called a {\em coding} (i.e., a letter-to-letter morphism).   

An {\em infinite word} over $A$ is a map $\mathbf{w}:\mathbb{N}\to A$. Note that the indexing of infinite words begins at position $0$ (which is quite convenient when dealing, for instance, with automatic sequences). A {\em factor} $u=u_1\cdots u_m$ of length $m$ occurring in a finite word $v=v_1\cdots v_n$ of length $n$ is a block of consecutive letters occurring in it, i.e., $m\le n$ and there exists $r\le n-m$ such that $u_j=v_{r+j}$ for $j\in\{1,\ldots,m\}$. In that case, we say that $u$ occurs in $v$ at position $j$. A {\em factor} of an infinite word $\mathbf{w}$ is a factor occurring in a finite prefix of $\mathbf{w}$. The set of factors (resp. the set of factors of length $n$) occurring in $\mathbf{w}$ is denoted by $\fac_\mathbf{w}$ (resp. $\fac_\mathbf{w}(n):=\fac_\mathbf{w}\cap A^n$). We denote similarly the set of factors of a finite word.

A morphism $f:A^*\to A^*$ is {\em prolongable} on the letter $a\in A$ if there exists a finite word $u$ such that $f(a)=au$ and if $\lim_{n\to +\infty}|f^n(a)|=+\infty$. In that case, the sequence $(f^n(a))_{n\ge 0}$ converges to an infinite word denoted by $f^\omega(a)$ that is said to be a {\em pure morphic word}. The image under a coding of a pure morphic word is said to be {\em morphic}. Let $k\ge 2$ be an integer. If the morphism $f:A^*\to A^*$ verifies $|f(a)|=k$ for all $a\in A$, then every infinite word of the form $g(f^\omega(a))$, where $g$ is a coding, is said to be $k$-automatic \cite{AlloucheShallit,Cobham}.

\section{General framework}\label{sec:gen}

Let $\sim$ be a reflexive and symmetric binary relation over $A^*$.
In many cases discussed in this survey, $\sim$ will be an equivalence relation (or even a congruence with respect to the concatenation of words). A trivial but useful example is given by the equality relation, each equivalence class is restricted to a singleton.

\begin{example}\label{exa:hamming}
    Let $k\ge 1$. Let $u,v$ be two words. We write $u\sim_{\mathcal{H},\le k}v$, if $|u|=|v|$ and the Hamming distance between $u$ and $v$ is at most $k$, i.e., 
$$d_{\mathcal{H}}(u,v):=\sum_{i=1}^{|u|} (1-\delta_{u_i,v_i})\le k$$
where $\delta_{a,b}=1$, if $a=b$; and $0$, otherwise. This relation is reflexive and symmetric but is not an equivalence relation. We have $abba\sim_{\mathcal{H},\le 1} abaa$, $abaa\sim_{\mathcal{H},\le 1} aaaa$ but $abba\not\sim_{\mathcal{H},\le 1} aaaa$.
\end{example}

A {\em language} over $A$ is a subset of $A^*$ (we only consider languages of finite words). The concatenation of words is naturally extended to the concatenation of languages: if $L,M$ are languages, $LM=\{uv\mid u\in L, v\in M\}$. Hence, the set $2^{A^*}$ of languages over $A$ equipped with concatenation is a monoid with $\{\varepsilon\}$ as neutral element.
\begin{definition}[growth function]
    Let $\sim$ be an equivalence relation over $A^*$ and $L\subset A^*$ be a language. We may consider the quotient $A^*/\!\sim$ and therefore the {\em growth function} of $L$ {\em with respect to} $\sim$ is defined as
$$\mathsf{g}_{\sim,L}:\mathbb{N}\to\mathbb{N},\ n\mapsto \#\bigl((L\cap A^n)/\!\sim\bigr).$$
If $\sim$ is the equality relation, $\mathsf{g}_{=,L}$ simply counts the number of words of length $n$ occurring in $L$. If $L=A^*$, then $\mathsf{g}_{\sim,A^*}(n)$ counts the number of equivalence classes of $\sim$ partitioning $A^n$.
\end{definition}

\begin{question}
Given an equivalence relation $\sim$ over $A^*$.
One can be interested in questions such as the following ones.
\begin{itemize}
  \item[Q.1.1] Compute or estimate the growth rate of $\mathsf{g}_{\sim,A^*}(n)$. 
  \item[Q.1.2] Given a specific language $L$, compute or estimate the growth order of $\mathsf{g}_{\sim,L}(n)$. Trivial bounds, for all $n\ge 0$, are given by $$1\le\mathsf{g}_{\sim,A^*}(n)\le \# (L\cap A^n)\le (\# A)^n.$$  
  \item[Q.1.3] For a class $\mathcal{F}$ of languages (e.g., the set of regular languages, the set of algebraic languages, the set of factors occurring in Sturmian words, etc.), does $\mathsf{g}_{\sim,L}$ have special properties for all $L\in\mathcal{F}$? Can it provide a characterization of $\mathcal{F}$?
\end{itemize}
\end{question}

\begin{example}
    If $L$ is a regular language (i.e., accepted by a finite automaton), then $(\mathsf{g}_{=,L}(n))_{n\ge 0}$ satisfies a linear recurrence equation with integer coefficients. This is a well-known consequence of Cayley--Hamilton theorem applied to the adjacency matrix of an automaton whose $n$th power counts walks of length $n$ between every pair of states.
\end{example}

As a special case of the previous definition, we can consider the {\em language of an infinite word}, i.e., the set of factors occurring in it.
\begin{definition}[complexity function]\label{def:complexity}
    Let $\sim$ be an equivalence relation over $A^*$. Since the quotient $A^*/\!\sim$ is well-defined, 
we thus define the {\em complexity function} of an infinite word $\mathbf{w}$ {\em with respect to} $\sim$ as
$$\mathsf{p}_{\sim,\mathbf{w}}:\mathbb{N}\to\mathbb{N},\ n\mapsto \#\bigl(\fac_\mathbf{w}(n)/\!\sim\bigr).$$
If $\sim$ is the equality relation, then $\mathsf{p}_{=,\mathbf{w}}$ is the usual {\em factor complexity} counting the number of factors of length $n$ occurring in $\mathbf{w}$ \cite{ELR}. The latter measure also leads to defining the {\em topological entropy} of $\mathbf{w}$. For a comprehensive presentation, see Cassaigne and Nicolas'~chapter \cite[Chap.~4]{CANT}. For instance, $\mathsf{p}_{=,\mathbf{w}}$ is in $\mathcal{O}(n)$ for every automatic sequence $\mathbf{w}$ \cite{Cobham,AlloucheShallit}. For a pure morphic word $\mathbf{w}$, a theorem of Pansiot \cite{Pan1,Pan2} shows that the growth order of $\mathsf{p}_{=,\mathbf{w}}$ can only take five forms $1,n,n\log\log n, n\log n,n^2$. See the survey \cite{Allouche93}.
\end{definition}

\begin{question}
Given an equivalence relation $\sim$ over $A^*$. 
One can be interested in questions such as the following ones.
\begin{itemize}
  \item[Q.2.1] Given a specific infinite word $\mathbf{w}$, compute or estimate the growth order of $\mathsf{p}_{\sim,\mathbf{w}}(n)$.
  \item[Q.2.2] For a class $\mathcal{F}$ of words (e.g., the set of Sturmian words \cite{Lot2}, the set of Arnoux--Rauzy words \cite{Arnoux}, the set of (pure) morphic words, automatic words, etc.), does $\mathsf{p}_{\sim,\mathbf{w}}$ have special properties for all $\mathbf{w}\in\mathcal{F}$? Can it provide a characterization of $\mathcal{F}$?
  \item[Q.2.3] A special case of the previous question is to study the set of words $\mathbf{w}$ such as $(\mathsf{p}_{\sim,\mathbf{w}}(n))_{n\ge 0}$ is bounded. Is there a Morse--Hedlund type result relating boundedness of the sequence $(\mathsf{p}_{\sim,\mathbf{w}}(n))_{n\ge 0}$ to the (ultimate) periodicity of the word $\mathbf{w}$ (see Theorem~\ref{the:mh} and Definition~\ref{def:periodicity} for generalizations of the concept of periodicity).
  \item[Q.2.4] Does $(\mathsf{p}_{\sim,\mathbf{w}}(n))_{n\ge 0}$ have a geometrical or a dynamical interpretation if $\mathbf{w}$ is derived from a dynamical system such as a coding of rotation? One can also be interested in arithmetical or algebraic interpretations when $\mathbf{w}$ is the expansion of a real number in a specific numeration system.
\end{itemize}
\end{question}

\begin{theorem}[Morse--Hedlund \cite{MH38}]\label{the:mh} An infinite word $\mathbf{w}$ is ultimately periodic, i.e., $\mathbf{w}=uvvv\cdots$ for some finite words $u,v$, if and only if the sequence $(\mathsf{p}_{=,\mathbf{w}}(n))_{n\ge 0}$ is bounded (by a constant). Otherwise stated, either $\mathbf{w}$ is ultimately periodic, or $\mathsf{p}_{=,\mathbf{w}}$ is increasing.
\end{theorem}

For a proof, for instance, see \cite[Section 4.3]{CANT} or \cite[Thm.~10.2.6]{AlloucheShallit}.

\begin{example}
    Sturmian words have been extensively studied \cite{Lot2,SS1} and several characterizations do exist. They can be defined
    as codings of particular rotations on the normalized interval $[0,1)$ with irrational angle $\alpha<1$
    when the interval $[0,1)$ is split into $[0,1-\alpha)$ and
    $[1-\alpha,1)$.  An infinite word $\mathbf{w}$ is Sturmian if and
    only if $\mathsf{p}_{=,\mathbf{w}}(n)=n+1$ for all $n\ge 0$. For rotation words obtained with another partition of the interval $[0,1)$, see \cite{berstelV,didier}. For the abelian equivalence $\sim_{\mathsf{ab}}$ discussed in the next section, Coven and Hedlund proved that an aperiodic word $\mathbf{w}$ is Sturmian if and only if $\mathsf{p}_{\sim_{\mathsf{ab}},\mathbf{w}}(n)=2$ for all $n\ge 1$ \cite{CovenHedlund}.
\end{example}

\begin{remark}
    If the equivalence relation $\sim$ is a {\em congruence} over $A^*$, i.e., for all $u_1,u_2,v_1,v_2$, if $u_i\sim v_i$ for $i=1,2$, then $u_1u_2\sim v_1v_2$, then the complexity function with respect to $\sim$ satisfies
$$\mathsf{p}_{\sim,\mathbf{w}}(m+n)\le \mathsf{p}_{\sim,\mathbf{w}}(m).\mathsf{p}_{\sim,\mathbf{w}}(n).$$
Indeed, every factor of length $m+n$ is the concatenation of a factor of length $m$ with a factor of length $n$ but the converse does not necessarily holds. The concatenation of two factors is not always a factor occurring in $\mathbf{w}$. 
\end{remark}

Since the works of Thue, the study of repetitions and unavoidable patterns is one of the cornerstones in combinatorics on words \cite{SS1,Thue1,Thue2}. An infinite word $\mathbf{w}\in A^\mathbb{N}$ {\em avoids} a set $S\subseteq A^*$, if $\fac_\mathbf{w}\cap S=\emptyset$. If such a word $\mathbf{w}$ exists, we say that  $S$ is {\em avoidable} over $A$. A set $S\subseteq A^*$ is {\em unavoidable} over $A$ whenever, for all $\mathbf{w}\in A^\mathbb{N}$, $\fac_\mathbf{w}\cap S\neq\emptyset$. 
We now introduce the notion of $\sim$-unavoidable pattern. For a survey on repetitions and avoidance, see Rampersad and Shallit's chapter \cite[Chap.~4]{CANT2}.

\begin{definition}[avoidance]
Let $B$ be a finite alphabet. Any finite word over $B$ will be called a {\em pattern}. Let $\sim$ be an equivalence relation over $A^*$. We now define a language-valued morphism that we will call a {\em substitution}\footnote{We here use the term `substitution' to avoid any confusion with the term `morphism'. In the literature, the word substitution is sometimes interchanged with morphism or non-erasing prolongable morphism.}. 
Let $h:B\to 2^{A^*}$ be a map satisfying
\begin{enumerate}
  \item for all $b\in B$, $h(b)$ is a non-empty set and $\varepsilon\not\in h(b)$;
  \item for all $b\in B$, if $u,v\in h(b)$, then $u\sim v$;
\end{enumerate}
Note that the image of every letter $b\in B$ is a subset of an equivalence class for $\sim$. The map $h$ is then extended to a morphism from $B^*$ to $2^{A^*}$ by setting $h(\varepsilon)=\{\varepsilon\}$ and $h(PQ)=h(P)h(Q)$ for all $P,Q\in B^*$. We say that the morphism $h$ is a {\em $\sim$-substitution}.  

Let $P\in B^*$ be non-empty. The pattern $P$ is {\em $\sim$-unavoidable} over $A$ if the language 
$$\mathcal{L}_\sim(P):=\bigcup_{\substack{h:B^*\to 2^{A^*}\\ h\text{ is a }\sim-\text{substitution}}} h(P)\subset A^*$$
is unavoidable over $A$. Otherwise,  $P$ is {\em $\sim$-avoidable} over $A$. If $\sim$ is the equality relation, we get back to the classical notion of avoidance. Every $=$-substitution is a non-erasing morphism and conversely. Note that it is enough to consider in the union defining $\mathcal{L}_\sim(P)$, the substitutions mapping letters of $B$ to equivalence classes of $\sim$. The set $\mathcal{L}_\sim(P)$ is called the {\em pattern language} associated with $P$ and $\sim$.
\end{definition}

\begin{definition}\label{def:square}
     Let $\sim$ be an equivalence relation over $A^*$. A {\em $\sim$-square} (resp. a {\em $\sim$-cube}) is a word in $\mathcal{L}_\sim(XX)$ (resp. in $\mathcal{L}_\sim(XXX)$), $X\in B$. In general, a {\em $\sim$-$n$th-power} is a word in $\mathcal{L}_\sim(X^n)$, $X\in B$, $n\in\mathbb{N}$.
\end{definition}

\begin{question}
    Let $\sim$ be an equivalence relation over $A^*$. 
    \begin{itemize}
      \item[Q.3.1] Given a pattern and an alphabet $A$ of size $k$, is this pattern $\sim$-avoidable over $A$? In particular, are $\sim$-squares or $\sim$-cubes avoidable? As an example, the Thue--Morse word avoids $=$-cubes or even overlaps corresponding to the pattern $XYXYX$. For a proof, for instance, see \cite{Lot}. 
      \item[Q.3.2]  Given a pattern that is $\sim$-avoidable, what is the minimal size of the alphabet such that it can be avoided?
      \item[Q.3.3] Given a pattern $P$, an alphabet $A$ of size $k$ and an integer $\ell$, does there exist an infinite word~$\mathbf{w}$ over $A$
$$\#\bigl( \fac_\mathbf{w} \,\cap\, \mathcal{L}_\sim(P) \bigr)\le \ell.$$
Note that this is Q.3.1 when $\ell=0$.
\item[Q.3.4] A reformulation of the previous question is to ask whether it exists an infinite word~$\mathbf{w}$ such that  
$$\fac_\mathbf{w}\cap \mathcal{L}_\sim(P) \subseteq A^{\le \ell}.$$
Otherwise stated, we only allow short occurrences of the pattern $P$.
\item[Q.3.5] Let $P$ be a pattern over $B$. A finite word $u\in A^*$ is {\em $\sim$-$P$-free}, if 
$$\fac_u\cap \mathcal{L}_\sim(P) =\emptyset.$$
A morphism $f:A^*\to A^*$ is {\em $\sim$-$P$-free} if, for all $\sim$-$P$-free words $u$, $f(u)$ also is $\sim$-$P$-free. Given a pattern $P$ and an alphabet $A$, does there exist a non-trivial prolongable $\sim$-$P$-free morphism? If such a morphism exists, then $P$ is $\sim$-avoidable over the alphabet $A$ \cite{Brandenburg}. As an example, the Thue--Morse morphism $a\mapsto ab$, $b\mapsto ba$ is overlap-free \cite{Lot}.
\item[Q.3.6] One can also be interested in enumeration questions such as counting the number of $\sim$-$P$-free finite words of length $n$. We give a few references where some interesting growth rates are exhibited \cite{Cassaigne,Kob,KS,CRam}. 
    \end{itemize}
\end{question}

\begin{remark}
    In Remark~\ref{rem:var}, a variant of $\sim$-$n$th-power is
    defined in the context of $\ell$-abelian equivalence. This
    definition could be extended to the general context presented here. Let $\sim$ be a binary relation over $A^*$. A word $u$ is a {\em strongly $\sim$-$n$th power} if there exists a `classical' $n$th power such that $u\sim v^n$. 

    There is also a notion of {\em approximated squares} introduced in
    \cite{Ochem}. As an example, a word of the form $uv$ with
    $u\sim_{\mathcal{H},\le k}v$ can be considered as an approximated
    square, with the relation defined in Example~\ref{exa:hamming}.
\end{remark}

\begin{example}
   Related to questions Q.3.3 and Q.3.4, Fraenkel and Simpson have built an infinite word over a $2$-letter alphabet with only 3 squares: $00$, $11$ and $0101$ \cite{Fraenkel94}. (It is easy to see that over a $2$-letter alphabet, any word of length at least $4$ contains a square.)
\end{example}

\begin{remark}
    The reader may also think about pattern matching. This topic will be considered, for two special cases ($\ell$-abelian equivalence and $k$-binomial equivalence), in Remarks~\ref{rem:pm1} and \ref{rem:pm2}.
\end{remark}

The following definition is inspired by the definition given in \cite{codes1} for similarity relations (see Section~\ref{sec:8}) and relational periods (in that case, the parameter $\ell$ is always equal to $1$). For a survey, see \cite[Chap. 6]{CANT2}. We will consider factorizations of an infinite word with words of a fixed length $\ell$ but one could relax this assumption. We are looking for a `period' made of $p$~words of length~$\ell$.

\begin{definition}[periodicity]\label{def:periodicity}
    Let $\sim$ be a reflexive and symmetric binary relation over $A^*$. 
Let $\mathbf{w}$ be an infinite word over $A$. Let $p,\ell \ge 1$ be integers.
    \begin{enumerate}
      \item The word $\mathbf{w}$ has $(p,\ell)$ as {\em global $\sim$-period} if there exists a sequence $(u_i)_{i\ge 0}$ of words of length $\ell$ such that $\mathbf{w}=u_0u_1u_2\cdots$ and, for all $i,j\in\mathbb{N}$, 
$$i\equiv j\pmod{p}\Rightarrow u_i\sim u_j.$$
\item The word $\mathbf{w}$ has $(p,\ell)$ as {\em external $\sim$-period} if there exist $p$ words $v_0,\ldots,v_{p-1}$ and a sequence $(u_i)_{i\ge 0}$ of words of length $\ell$ such that $\mathbf{w}=u_0u_1u_2\cdots$ and, for all $n\in\mathbb{N}$ and all $r\in\{0,\ldots,p-1\}$, $u_{np+r}\sim v_r$.  
\item The word $\mathbf{w}$ has $(p,\ell)$ as {\em local $\sim$-period} if there exists a sequence $(u_i)_{i\ge 0}$ of words of length $\ell$ such that $\mathbf{w}=u_0u_1u_2\cdots$ and, for all $i\ge 0$, $u_i\sim u_{i+p}$.
     \end{enumerate} 
If such a pair $(p,\ell)$ exists, we say that $\mathbf{w}$ is {\em globally} (resp. {\em externally}, {\em locally}) {\em $\sim$-periodic} and $(p,\ell)$ is a  {\em global} (resp. {\em external}, {\em local}) {\em $\sim$-period}.
\end{definition}

\begin{example}
Let $u\sim_{\mathcal{H},\le 1}v$  be the relation defined in Example~\ref{exa:hamming}.
Consider the generalized Thue--Morse word (OEIS {\tt A004128})\footnote{Let $m\ge 2$ and $k\ge 2$ be integers. The infinite word $\mathbf{t}_{k,m}:=(s_k(n)\mod{m})_{n\ge 0}$ over the alphabet $\{0,\ldots,m-1\}$, where $s_k(n)$ is the sum-of-digits of the base-$k$ expansion of $n$, is overlap-free if and only if $k\le m$. It is also known that $\mathbf{t}_{k,m}$ contains arbitrarily long squares \cite{AlloucheShallit2}.} over $\{0,1,2\}$ 
$$012 120 201 120 201 012 201 012 120 120 012 201 \cdots$$
and apply the morphism $0\mapsto aaaa$, $1\mapsto abaa$ and $2\mapsto abba$ to get the word
$$\mathbf{w}=aaaaabaaabba abaaabbaaaaa abbaaaaaabaa abaaabbaaaaa abbaaaaaabaa  \cdots$$
It has external period $(1,4)$, for all $n\ge 0$, $w_{4n}w_{4n+1}w_{4n+2}w_{4n+3}\sim_{\mathcal{H},\le 1} abaa$.
\end{example}

\begin{remark}
    In the previous definition, if $\sim$ is also transitive, i.e., $\sim$ is an equivalence relation, then the three notions of global, external and local $\sim$-periods coincide. In that case, we simply say that a word is $\sim$-periodic or {\em ultimately} $\sim$-periodic if it has a $\sim$-periodic suffix.
\end{remark}

\begin{lemma}
    Let $\sim$ be a congruence over $A^*$. If $\mathbf{w}$ has the pair $(p,\ell)$ as $\sim$-period, then $\mathbf{w}$ has $(1,p\ell)$ as $\sim$-period
\end{lemma}

\begin{proof}
    The exists a sequence $(u_i)_{i\ge 0}$ of words of length $\ell$ such that $\mathbf{w}=u_0u_1u_2\cdots$ and, for all $n\in\mathbb{N}$ and $r\in\{0,\ldots,p-1\}$, $u_{np+r}\sim u_r$. Since $\sim$ is a congruence, for all $n\in\mathbb{N}$, 
$u_{np} u_{np+1}\cdots u_{np+p-1}\sim u_0\cdots u_{p-1}$. Thus we can consider the sequence $(u_{np} u_{np+1}\cdots u_{np+p-1})_{n\ge 0}$ of words of length $p\ell$ showing that $\mathbf{w}$ is $(1,p\ell)$-periodic.
\end{proof}

\begin{question}
    Given a (pure) morphic $\mathbf{w}$ (or a word given with a finite description) and a relation $\sim$, is it decidable whether or not $\mathbf{w}$ is of the form $u\mathbf{x}$ where $u$ is a finite word and $\mathbf{x}$ is  globally (resp. externally, locally) $\sim$-periodic? See, for instance, \cite{HarjuLinna,Pansiot86,Durand,KarkiR}.
\end{question}

\begin{remark}
    Other periodicity-related topics such as variants of Fine--Wilf theorem
    \cite{berboa,BS3,fw1,fw2,fw3,fw4,fw5} or codes and defect effect \cite{codes1,codes2} may be
    considered.
\end{remark}

We introduce the last concept of this part of the paper. A subset $X=\{x_0<x_1<x_2<\cdots\}\subseteq\mathbb{N}$ is {\em syndetic} (or, {\em with bounded gaps}) if there exists a constant $C$ such that $x_{i+1}-x_i<C$ for all $i\ge 0$. In the last part of this section, we assume that if $u\sim v$, then $|u|=|v|$.

\begin{definition}[uniform recurrence]
Let $\sim$ be a reflexive and symmetric binary relation over $A^*$.    For every $u\in\fac_\mathbf{w}$, consider the set of positions where occurs a factor in relation with $u$
$$\mathsf{Occ}_{\sim,u}(\mathbf{w}):=\{i\ge 0\mid v_i\cdots v_{i+|u|-1} \sim u\}$$
If for all $u\in\fac_\mathbf{w}$, the set $\mathsf{Occ}_{\sim,u}(\mathbf{w})$ is infinite (resp. infinite and syndetic), then we say that $\mathbf{w}$ is {\em $\sim$-recurrent} (resp. {\em $\sim$-uniformly recurrent}).
\end{definition}

\begin{definition}\label{def:return}
If $\mathbf{w}$ is $\sim$-uniformly recurrent, then we can factorize the word $\mathbf{w}$ using the set of positions $\mathsf{Occ}_{\sim,u}(\mathbf{w})=\{i_1<i_2<\cdots\}$: $$\mathbf{w}=(w_0\cdots w_{i_1-1})(w_{i_1}\cdots w_{i_2-1})(w_{i_2}\cdots w_{i_3-1})\cdots$$
Observe that uniform recurrence implies that the set of words $\{ w_{i_j}\cdots w_{i_{j+1}-1}\mid j\ge 1\}$ is finite. These words are called the {\em $\sim$-return words} to $u$. Each such word shares a common prefix with a word in relation with $u$ for $\sim$. If it is longer than $u$ then it has $u'$ as a prefix for some $u'\sim u$. This notion is similar to the first return map in dynamical systems theory.
\end{definition}

\begin{example}
    Consider the Thue--Morse word (OEIS {\tt A010060}) and the abelian equivalence $\sim_{\mathsf{ab}}$ defined in the next section. With the prefix $01101$, we have marked all the occurrences of a factor of length $5$ having precisely $3$ ones, i.e., that is a rearrangement (or anagram) of this prefix:
$$|\underbrace{0}_{1}|\underbrace{110}_2|\underbrace{100}_3|110|0|\underbrace{1}_4|0|\underbrace{11010}_5|0|\underbrace{10}_6|1|10|0|110|100|110|0|10|1|10|0|1101001\cdots$$
One can prove that the only factors that occur are $\{0,1,10,100,110,11010\}$ mapping this set onto $\{1,\ldots,6\}$ (where the usual convention is that the index is given by the order of first appearance of the factor within the factorized word), we can code the previous factorization by
$$123214151646123216461\cdots$$
Such a sequence is called a {\em derived sequence} and is denoted by $\mathsf{D}_{\sim_{\mathsf{ab}},01101}(\mathbf{w})$
\end{example}

This concept of derived sequence (or descendant) was introduced independently by Durand and Holton and Zamboni \cite{HoltonZamboni}. A morphism $f:A^*\to A^*$ is {\em primitive} if the matrix $M=(|f(a)|_b)_{a,b\in A}\in\mathbb{N}^{A\times A}$ is primitive, i.e., there exists $n$ such that $M^n>0$.

\begin{theorem}\cite{Durand:1998b}
    An infinite uniformly recurrent word $\mathbf{w}$ is of the form $g(f^\omega(a))$ where $g:A^*\to B^*$ is a coding and $f:A^*\to A^*$ is a primitive morphism prolongable on $a$ if and only if the set $\{\mathsf{D}_{=,p}(\mathbf{w})\mid p\text{ is a prefix of }\mathbf{w}\}$ is finite.
\end{theorem}

\begin{prop}\cite[Prop.~5.1]{Durand:1998b}
    Let $f$ be a primitive morphism prolongable on the letter $a$. For every prefix $p\neq\varepsilon$ of $f^\omega(a)$, the sequence $\mathsf{D}_{=,p}(f^\omega(a))$ is also the fixed point of a primitive morphism.
\end{prop}


\section{Abelian framework}\label{sec:a1}

Erd\H{o}s raised the question whether abelian squares can be avoided by an infinite word over an alphabet of size $4$.  We refer to the paper \cite{Erdos1} that can easily be accessed\footnote{It is common to refer to another Erd\H{o}s'paper: Some unsolved problems, {\em Magyar Tud. Akad. Mat. Kutat\'o Int. K\"ozl.} {\bf 6} (1961), 221--254.}, the last problem of the list of $28$ problems is the following: ``{\em Let $N(k)$ be the least number $N$ with the property that each sequence $\{s_n\}_{n=1}^N$ of numbers taken from the set $\{1,\ldots,k\}$ contains two adjacent blocks such that each is a rearrangement of the other. My earliest conjecture, that $N(k)=2^k-1$, has been disproved by Bruijn and myself. It is not even known whether $N(4)<\infty$.}'' (Exhausting all the possible cases, it is an easy exercise to prove that any long enough finite word over an alphabet of size $3$ contains an abelian square.)

\begin{definition}
Let $A=\{1<\cdots <k\}$ be a finite alphabet that is assumed to be ordered. We consider the {\em abelianization map} (also called {\em Parikh map}, see Theorem~\ref{the:parikh}) denoted by $\Psi:A^*\to\mathbb{N}^k$. It is a morphism of monoids where $\Psi(u)=(|u|_1,\ldots,|u|_k)^\mathsf{T}$ for all $u\in A^*$. Indeed, $\Psi(uv)=\Psi(u)+\Psi(v)$ for all $u,v\in A^*$. In particular, if $\Psi$ is extended to languages, $\Psi^{-1}(\Psi(L))$ is the {\em commutative closure} of the language $L$.
\end{definition}

\begin{definition}
The notion of abelian square introducing this section is a special case of Definition \ref{def:square} when considering the {\em abelian equivalence} $\sim_{\mathsf{ab}}$ over $A^*$ defined by 
$$u \sim_{\mathsf{ab}} v \Leftrightarrow \Psi(u)=\Psi(v).$$
Otherwise stated, $u$ is obtained by applying a permutation to the letters of $v$. The relation $\sim_{\mathsf{ab}}$ is clearly a congruence.
\end{definition}

About Definition \ref{def:complexity}, one introduces the notion of {\em abelian complexity} $\mathsf{p}_{\sim_{\mathsf{ab}},\mathbf{w}}$ where factors occurring in $\mathbf{w}$ are counted up to abelian equivalence.   In contrast with the usual factor complexity function $\mathsf{p}_{=,\mathbf{w}}$ which is non-decreasing, this property no longer holds for  $\mathsf{p}_{\sim_\mathsf{ab},\mathbf{w}}$: it is possible that $\mathsf{p}_{\sim_\mathsf{ab},\mathbf{w}}(n)>\mathsf{p}_{\sim_\mathsf{ab},\mathbf{w}}(n+1)$ for some $n$. For instance, for the Tribonacci word~$\mathbf{t}$ (OEIS {\tt A000073}) $\mathsf{p}_{\sim_\mathsf{ab},\mathbf{t}}(7)=4$ but $\mathsf{p}_{\sim_\mathsf{ab},\mathbf{t}}(8)=3$. A few references are \cite{Ric1,Ric2,BS1,Tur1,Tur2} and \cite{RampersadPF} where the abelian complexity of the paper-folding word is shown to be $2$-regular (in the sense of Allouche and Shallit), see, for instance, \cite{BerstelReutenauer}. In particular, bounded abelian complexity is related to balance properties and existence of frequencies \cite{Adam2003}.

\begin{theorem}\cite{Ric2}
An infinite word has a bounded abelian complexity if and only if it is $C$-balanced for some $C>0$, i.e., for all $u,v\in\fac_\mathbf{w}(n)$, $n\ge 1$, we have $|\, |u|_a-|v|_a\, |\le C$ for every letter $a$ in the alphabet.  
\end{theorem}

\begin{properties}
Ker{\"a}nen has built a pure morphic word  over a $4$-letter alphabet that avoids abelian squares \cite{Keranen,Brown}. Dekking has obtained an infinite word over a $3$-letter alphabet that avoids abelian cubes, and an infinite word over a $2$-letter alphabet that avoids abelian $4$-powers \cite{Dekking}. (Note that in all these results, the size of the alphabet is optimal.)
\end{properties}

About {\em abelian power-free morphisms}, see \cite{Carpi,CR2}. See also \cite{CaRa}. 

\begin{properties}
 Every infinite word over a $2$-letter alphabet contains arbitrarily long abelian squares and there exists an infinite word that avoids squares of the form $uu'$ with $u\sim_{\mathsf{ab}}u'$ and $|u|\ge 3$ \cite{Entringer}.   
\end{properties}

About enumeration results like counting the number of finite words of length $n$ avoiding abelian cubes, see \cite{Aberkane,Carpi2}. 
\medskip

On the characterization of classes of words with respect to abelian equivalence and in particular, Sturmian words. Let us mention the following results. Extending a result of Vuillon in \cite{Vuillon} to $\sim_{\mathsf{ab}}$. (recall Definition~\ref{def:return} of return words.) 

\begin{theorem}\cite{PuZ}
    A recurrent infinite word is Sturmian if and only if each
of its factors has two or three $\sim_{\mathbf{ab}}$-return words.
\end{theorem}

\begin{properties}\cite{RSV}
    \begin{enumerate}
      \item Let $\mathbf{w}$ be a recurrent word. The set $\sim_{\mathbf{ab}}$-return words is finite if and only if $\mathbf{w}$ is periodic.
      \item Let $\mathbf{w}$ be a Sturmian word (we assume that the notion of intercept is understood). The set of $\sim_{\mathbf{ab}}$-return words to the prefixes is finite if and only if $\mathbf{w}$ has a non-zero intercept.
    \end{enumerate}
\end{properties}
 The latter result can be extended to rotation words \cite{RRS}. 

 \begin{remark}
     Closely related to abelian equivalence, one can also consider an {\em additive relation} where two words $u,v$ (one can add the extra assumption that $|u|=|v|$) over a finite alphabet of integers are {\em additively equivalent}, if $\sum u_i=\sum v_i$. For instance, $134233$ is an additive square. The paper \cite{CasSha} shows the existence of an infinite word over $\{0,1,3,4\}$ avoiding additive cubes (OEIS {\tt A191818}). Also see \cite{Rao15} where subsets of $\mathbb{N}$ of size $3$ are considered.
 \end{remark}


\section{$k$-abelian equivalence}\label{sec:a2}
We now present a first generalization of the concept of abelian equivalence stemming from a classical result in formal language theory: Parikh's theorem. See any standard textbook on formal language theory, e.g. \cite{Sudkamp}, in particular for the definition of a context-free language. A set $M\subseteq\mathbb{N}^d$ is said to be {\em linear}, if there exist
$x\in\mathbb{N}^d$ and a finite set (possibly empty) $V=\{v_1,\ldots,v_k\}\subset\mathbb{N}^d$
such that $$M=\left\{x +\sum_{i=1}^k \lambda_i\, v_i\mid \lambda_1,\ldots,\lambda_k\in \mathbb{N}\right\}.$$ A
finite union of linear sets is a {\em semi-linear} set.

\begin{theorem}[Parikh's theorem \cite{Parikh}]\label{the:parikh}
    If $L$ is a context-free language over a $k$-letter alphabet, then $\Psi(L)$ is a semi-linear set of $\mathbb{N}^k$.
\end{theorem}

Let $\ell\ge 1$. Trying to strengthen Parikh's theorem, instead of counting occurrences of letters, we could
count occurrences of factors of length at most $\ell$ \cite{Karhu80}. In that setting, assuming that $A=\{1<\cdots <k\}$ is ordered, we get extra information on the structure of the word given by an extended abelianization 
map, also called {\em generalized Parikh mapping},
$$\Psi_\ell:A^*\to \mathbb{N}^{k+k^2+\cdots +k^\ell}$$
where, for all $u\in A^*$, 
\begin{equation}
    \label{eq:psil}
\Psi_\ell(u)=(|u|_1,\ldots,|u|_k,|u|_{11},\ldots,|u|_{kk},\ldots,|u|_{1^\ell},\ldots,|u|_{k^\ell})
\end{equation}
and $|u|_v$ denotes the number of occurrences of the factor $v$ in $u$. Note that the size of $\Psi_\ell(u)$ grows exponentially with $\ell$: it is a vector of size $k(k^\ell-1)/(k-1)$. As an example, $|0110100|_{10}=2$ and $|01110|_{11}=2$ (overlaps are allowed). The following relation was introduced in \cite{KSZ}.

\begin{definition}\label{def:kabelian} Let $\ell\ge 1$ be an integer.
    Two finite words $u$ and $v$ are {\em $\ell$-abelian equivalent}, if $\Psi_\ell(u)=\Psi_\ell(v)$. We write $u\sim_{\ell-\mathsf{ab}}v$. Otherwise stated, if, for all words $x\in A^{\le \ell}$, $|u|_x=|v|_x$. Clearly, for $\ell=1$ we are back to the usual abelian equivalence.
\end{definition}
Note that, for all $n\le |u|$, $$|u|=\sum_{x\in A^n} |u|_x+n-1$$
and $\Psi_\ell(A^*)$ is a strict subset of $\mathbb{N}^{k(k^\ell-1)/(k-1)}$. 
\begin{example}
The words $u=010110$ and $v=011010$ are $3$-abelian equivalent. We have $|u|_0=3=|v|_0$, $|u|_1=3=|v|_1$, $|u|_{00}=0=|v|_{00}$, $|u|_{01}=2=|v|_{01}$, $|u|_{10}=2=|v|_{10}$,  $|u|_{11}=1=|v|_{11}$. Finally, $|u|_{010}=1=|v|_{010}$,  $|u|_{101}=1=|v|_{101}$, $|u|_{011}=1=|v|_{011}$ $|u|_{110}=1=|v|_{110}$.
 But the two words $u$ and $v$ are not $4$-abelian equivalent: the factor $1010$ occurs in $v$ but not in $u$. The relation $\sim_{(\ell+1)-\mathsf{ab}}$ is a refinement of $\sim_{\ell-\mathsf{ab}}$ (see the lattice in Figure~\ref{fig:lattice}).
\end{example}
\begin{remark}\label{rem:rational}
    In terms of rational series (we refer the reader to \cite{BerstelReutenauer} for definitions), since the characteristic series of $A^*$ denoted by $\underline{A}^*$ is rational, we deduce that the formal series in $\mathbb{N}\langle\langle A\rangle\rangle$ 
$$\underline{A}^*u\underline{A}^*=\sum_{w\in A^*}|w|_u\, w$$
is rational. 
\end{remark}
It is not difficult to see that two words $u$ and $v$ of length at least $\ell-1$ are $\ell$-abelian equivalent if and only if they share respectively the same prefix and the same suffix of length $\ell-1$ and if $|u|_x=|v|_x$ for all words $x$ of length $\ell$. This property implies that $\sim_{\ell-\mathsf{ab}}$ is again a congruence. In \cite{KSZ}, the growth of $\mathsf{g}_{\sim_{\ell-\mathsf{ab}}}$ is estimated. Ultimately periodic words and Sturmian words can be characterized by the $\ell$-abelian complexity function.

\begin{theorem}\cite{KSZ}
    Let $\ell\ge 1$. An infinite aperiodic word is Sturmian if and only if 
$$\mathsf{p}_{\sim_{\ell-\mathsf{ab}},\mathbf{w}}(n)=\left\{
    \begin{array}{ll}
        n+1,&\text{ if }n<2\ell;\\
        2\ell,&\text{ if }n\ge 2\ell.\\
    \end{array}\right.$$
\end{theorem}

About the fluctuations of $\mathsf{p}_{\sim_{\ell-\mathsf{ab}},\mathbf{w}}$, see the papers \cite{CaSa,KSZ2}. The $2$-abelian complexity of the Thue-Morse word is shown to be $2$-regular in \cite{Parreau} and, independently, in \cite{Gre}.
\medskip

Many results on avoidance are available. In \cite{Rao15}, Rao provides morphic words avoiding $\ell$-abelian powers: an infinite word over a $2$-letter alphabet avoiding $2$-abelian cubes and an infinite word over a $3$-letter alphabet avoiding $3$-abelian squares. The paper also deals with bounds on enumeration results in that context of avoidance. About other avoidance results, also see \cite{Huo1,Huo2,Huo3}

\begin{remark}\label{rem:var}
    A variant of the notion of repetition is considered in
    \cite{Huo4}, a word is a {\em strongly $\ell$-abelian $n$th power}, if it
    is $\ell$-abelian equivalent to a `classical' $n$th power. As an example, the word $aabb$ is not an abelian square because $aa\not\sim_{\mathsf{ab}}bb$ but it is a strongly abelian square because $aabb\sim_{\mathsf{ab}}(ab)(ab)$.
\end{remark}

\begin{remark}\label{rem:pm1}[$\ell$-abelian pattern matching]
    Pattern matching has many applications, here we concentrate on `approximate' pattern matching problems (that can be considered with respect to a given equivalence relation). In \cite{patmat}, making use of suffix arrays, the following problems are positively answered.
\begin{itemize}
         \item Given $\ell\ge 1$ and two words $u,v$ of length $n$,
           decide, in polynomial time with respect to $n$ and $\ell$,
           whether or not $u\sim_{\ell-\mathsf{ab}} v$.
       \item Given $\ell\ge 1$ and two words $w,x$, find, in polynomial
         time, all occurrences of factors of $w$ which are
         $\ell$-abelian equivalent to $x$.
          \item Given two $u,v$ of length $n$, find the largest $\ell$
            such that $u\sim_{\ell-\mathsf{ab}} v$.
    \end{itemize}
\end{remark}


\section{Binomial coefficients}\label{sec:a4}
The notion of a binomial coefficient of words is classical in combinatorics on words. See, for instance, Sakarovitch and Simon's chapter in \cite{Lot}. Let $w,x\in A^*$. The integer denoted by $$\binom{w}{x}$$ 
counts the number of times $x$ appears as a (scattered) subword\footnote{This is the reason why we make a distinction between factors made of consecutive letters and subwords. Be aware that in the literature these two terms are sometimes used with the same meaning.} of $w$, i.e., $x$ occurs as a subsequence of $w$. Otherwise stated, we count the number of increasing maps $\varphi:\{1,\ldots,|x|\}\to\{1,\ldots,|w|\}$ such that 
$$\varphi(1)<\cdots<\varphi(|x|)\quad \text{ and }\quad w_{\varphi(1)}\cdots w_{\varphi(|x|)}=x.$$
As an example, we have $\binom{aabbab}{ab}=7$. It generalizes the usual binomial coefficients of integers because, over a $1$-letter alphabet,
$$\binom{a^m}{a^n}=\binom{m}{n},\quad m,n\in\mathbb{N}.$$
These coefficients can easily be computed from the relations 
$$\binom{w}{\varepsilon}=1,\quad\quad \binom{w}{x}=0,\ \text{ if }|w|<|x|$$
and
$$\forall u,v\in A^*, a,b\in A,\quad \binom{ua}{vb}=\binom{u}{vb}+\delta_{a,b}\binom{u}{v}.$$
\begin{remark}
We have an observation similar to Remark~\ref{rem:rational}.  Let $u=u_1\cdots u_n$. 
    In terms of rational series (we again refer to \cite{BerstelReutenauer}), since the characteristic series of $A^*$ is rational, we deduce that the formal series in $\mathbb{N}\langle\langle A\rangle\rangle$ 
$$\underline{A}^*u_1\underline{A}^* u_2\underline{A}^*\cdots \underline{A}^*u_n\underline{A}^*=\sum_{w\in A^*}\binom{w}{u}\, w$$
is rational. 
\end{remark}
It is not difficult to prove the following result.
\begin{prop}\label{pro:bincon}
    Let $s,t,w$ be three words of $A^*$. Then we have
    $$\binom{sw}{t}=\sum_{uv=t} \binom{s}{u}\binom{w}{v}.$$
\end{prop}
Let us mention the so-called {\em Cauchy inequality}. Several proofs of this result exist, see \cite{Salomaa2003}.
\begin{theorem}
For all words $w,x,y,z\in A^*$, we have
    $$\binom{w}{y}\binom{w}{xyz}\le \binom{w}{xy}\binom{w}{yz}.$$
\end{theorem}

A general question is to `reconstruct' a word from some of its binomial coefficients: What numbers $\binom{w}{u}$ suffice to
determine the word $w$ uniquely? See, for instance, \cite{Salomaa2005}. Sch\"utzenberger and Simon proved that two words of length $n$ with the same subwords of length up to $\lfloor n/2\rfloor+1$ are identical. In \cite{krasikov}, it is shown that any word of length $n$ is uniquely determined by all its subwords of length $k$, if $k\ge \lfloor 16\sqrt{n}/7\rfloor+5$.  The authors relate this problem to well-known vertex reconstruction problems in graph theory and trace the origin of the problem back to \cite{kala}. For algorithmic considerations, see, for instance, \cite{Dress}.
\medskip

Similarly to $\ell$-abelian equivalence, these binomial coefficients allows us to define an independent refinement of abelian equivalence.
\begin{definition}
  Let $k\ge 1$. Two words $u,v$ are {\em $k$-binomially equivalent}, and we write $u\sim_{k-\mathsf{bin}} v$, if and only if
$$\binom{u}{x}=\binom{v}{x}\quad \forall x\in A^{\le k}.$$
In particular, since $\binom{w}{a}=|w|_a$, for $a\in A$, if $k=1$, then we have the usual equivalence relation $\sim_\mathsf{ab}$. The fact that $\sim_{k-\mathsf{bin}}$ is a congruence is a consequence of Proposition~\ref{pro:bincon}. Similarly to \eqref{eq:psil},  assuming that $A=\{1<\cdots <k\}$ is ordered, one could introduce the map
$$\Psi_k'(u):=\left(\binom{u}{1},\ldots,\binom{u}{k},\binom{u}{11},\ldots,\binom{u}{kk},\ldots,\binom{u}{1^\ell},\ldots,\binom{u}{k^\ell}\right)$$
and $u,v$ are $k$-binomially equivalent if and only if $\Psi_k'(u)=\Psi_k'(v)$. In \cite{salo2010}, the $2$-binomial equivalence was called {\em binary equivalence}.
\end{definition}

As observed in \cite{dudik}, if $|u|\ge k\ge |x|$ , then 
$$\binom{|u|-|x|}{k-|x|} \binom{u}{x} = \sum_{t\in A^k} \binom{u}{t}\binom{t}{x}.$$
Indeed, on the right hand side, a fixed occurrence of the subword $x$ in $u$ is counted as many times as it appears in any bigger subword. Thus, if the positions of the letters of $x$ are fixed, we can build a bigger subword made of $k$ symbols and containing that particular occurrence of $x$ by selecting $k-|x|$  positions amongst the $|u|-|x|$ remaining ones in $u$. Consequently, we deduce the following result.
\begin{lemma}
If $u,v$ are words of length at least $k$, then $u\sim_{k-\mathsf{bin}} v$, if and only if $\binom{u}{t}=\binom{v}{t}$ for all words $t$ of length $k$.
\end{lemma}

\begin{example}\label{exa:kbin}
    The four words $ababbba$, $abbabab$, $baabbab$ and
    $babaabb$ are $2$-binomially equivalent. For any $w$ amongst these words,
    we have the following coefficients
    $$\binom{w}{aa}=3,\ \binom{w}{ab}=7 ,\ \binom{w}{ba}=5 ,\ \binom{w}{bb}=6.$$
But one can check that they are not $3$-binomially equivalent, as an example,
$$\binom{ababbba}{aab}=3 \text{ but }\binom{abbabab}{aab}=4$$
indeed, for this last binomial coefficient, $aab$ appears as subwords
$w_1w_4w_5$, $w_1w_4w_7$, $w_1w_6w_7$ and $w_4w_6w_7$. 
The $k$-abelian equivalence and the $k$-binomial equivalence relations are incomparable.  Considering
again the first two words, we find $|ababbba|_{ab}=2$ and
$|abbabab|_{ab}=3$, showing that these two words are not $2$-abelian
equivalent. Conversely, the words $abbaba$ and $ababba$ are
$2$-abelian equivalent but are not $2$-binomially equivalent:
$$\binom{abbaba}{ab}=4\text{ but }\binom{ababba}{ab}=5.$$
\end{example}

\begin{remark}
    Since The relation $\sim_{(k+1)-\mathsf{bin}}$ is a refinement of $\sim_{k-\mathsf{bin}}$, we have a lattice of relations over $A^*$ as depicted in Figure~\ref{fig:lattice}. The coarsest relation is abelian equivalence and the finest relation is equality. The relations $\sim_{\psi_{w_0\cdots w_k}}$ will be introduced in Definition~\ref{def:newpsi}.
\begin{figure}[h!tb]
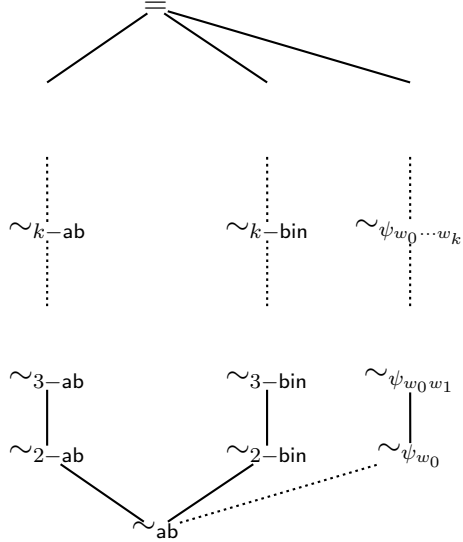

    \centering
\begin{psmatrix}[mnode=r,colsep=0.6,rowsep=0.5]
&[name=1] $=$ & & \\
[name=1a] && [name=1b]& [name=q4] \\
[name=5a] && [name=5b]& [name=q3] \\
[name=4a] $\sim_{k-\mathsf{ab}}$& & [name=4b] $\sim_{k-\mathsf{bin}}$ & [name=q2] $\sim_{\psi_{w_0\cdots w_k}}$\\
[name=6a] && [name=6b] & [name=q1]\\
[name=3a] $\sim_{3-\mathsf{ab}}$& & [name=3b] $\sim_{3-\mathsf{bin}}$ & [name=p2] $\sim_{\psi_{w_0w_1}}$\\
[name=2a] $\sim_{2-\mathsf{ab}}$& & [name=2b] $\sim_{2-\mathsf{bin}}$ & [name=p1] $\sim_{\psi_{w_0}}$ \\
& [name=ab] $\sim_{\mathsf{ab}}$ & & \\
\ncline{1}{1a}\ncline{1}{1b}\ncline{1}{q4}
\ncline[linestyle=dashed, dash=1pt 1.5pt]{4a}{5a}\ncline[linestyle=dashed, dash=1pt 1.5pt]{4b}{5b}

\ncline[linestyle=dashed, dash=1pt 1.5pt]{4a}{6a}\ncline[linestyle=dashed, dash=1pt 1.5pt]{4b}{6b}
\ncline{ab}{2a}\ncline{ab}{2b}
\ncline{3a}{2a}\ncline{3b}{2b}
\ncline{p1}{p2}
\ncline[linestyle=dashed, dash=1pt 1.5pt]{q1}{q2}\ncline[linestyle=dashed, dash=1pt 1.5pt]{q3}{q2}
\ncline[linestyle=dashed, dash=1pt 1.5pt]{ab}{p1}
\end{psmatrix}
\caption{A lattice of congruences, the finest is $=$, the coarsest is $\sim_{\mathsf{ab}}$.}
    \label{fig:lattice}
\end{figure}
\end{remark}

In the literature, one also finds the notion of {\em $k$-spectrum} of a word $u$ which is the (formal) polynomial (we refer to \cite{BerstelReutenauer} for definitions) in $\mathbb{N}\langle A^*\rangle$ of degree $k$ $$\mathsf{Spec}_{u,k}:=\sum_{w\in A^{\le k}} \binom{u}{w}\, w.$$
The $2$-spectrum of the word $u=abbab$ is
$$\mathsf{Spec}_{u,2}=1\varepsilon+2a+3b+aa+4ab+2ba+3bb.$$
If we replace $a$ with $0$ and $b$ with $1$ and if every word is preceded by a leading $1$, every word $w$ over $\{a,b\}$ corresponds to a unique integer (there is no leading $0$) written in base-$2$, $\val_2(1w)$, so this spectrum can also be represented as a univariate polynomial were the word $w$ is replaced with $X^{\val_2(1w)}$. With the same word $u$, we have
\begin{equation}
    \label{eq:poly}
    1+2X^2+3X^3+X^4+4X^5+2X^6+3X^7.
\end{equation}
The $3$-spectrum of the same word $u$ is 
$$\mathsf{Spec}_{u,3}=\mathsf{Spec}_{u,2}+aab+2aba+3abb+2bab+bba+bbb.$$
Note that, just like $\Psi_k'$, the $k$-spectrum grows exponentially with $k$, it contains
$(\# A^{k+1}-1)/(\# A-1)$ (possibly zero) coefficients.  Let us
quote Salomaa: ``{\em a notion often mentioned but not much investigated
in the literature, \cite{SS1,SS2,SS3,Salomaa2003}, is that of a $t$-spectrum.}''. In particular, in terms of `reconstruction' what is the relation between $n$ and $k$ such that if two words $u,v$ of length $n$ have the same $k$-spectrum, then $u=v$?

\begin{remark}
    Two words are $k$-binomially equivalent if and only if they have the same $k$-spectrum.
\end{remark}

\begin{properties}
About avoidance of $k$-binomial repetitions, see \cite{RaoRS}: $2$-binomial
squares (resp. cubes) are avoidable over a $3$-letter (resp. $2$-letter) alphabet. The sizes of the respective alphabets are optimal.
\end{properties}

\begin{theorem}\cite[Thm.~7]{RS}
    If $\mathbf{x}$ is a Sturmian word, then $\mathsf{p}_{\sim_{2-\mathsf{bin},\mathbf{x}}}(n)=n+1$ for all $n\ge 0$. In particular, we also have  $\mathsf{p}_{\sim_{k-\mathsf{bin},\mathbf{x}}}(n)=n+1$ for all $n\ge 0$ and $k\ge 2$.
\end{theorem}

\begin{properties}\cite[Thm.~17]{RS}
    If an infinite recurrent word has bounded $2$-binomial complexity, then the frequency of each symbol exists and is rational.
\end{properties}

\begin{lemma}
Let $k\ge 1$. If $x\sim_{k-\mathsf{bin}}y$, then 
$pxqyr\sim_{(k+1)-\mathsf{bin}}pyqxr$.    
\end{lemma}

\begin{proof}
    If we have to count the number of occurrences of a subword $z$ of
    length at most $k+1$, either this subword occurs completely inside
    one of the factors $p,q,r,x,y$, or it is obtained from several
    shorter subwords occurring in at least two of these factors. To
    get the conclusion, observe that, by assumption, $x,y$ share
    the exactly the same subwords of length at most $k$.
\end{proof}

But it is not clear that a form of converse for this result exists
(see Proposition~\ref{pro:salo} where we have a characterization of
equivalent words in terms of this kind of transformations but only for
a $2$-letter alphabet). Over a $3$-letter alphabet:
$2100221\sim_{2-\mathsf{bin}} 0221102$ but $2100221$ cannot be
factorized into $pxqyr$ with $x\sim_{\mathsf{ab}} y$ and $x\neq y$.

\begin{remark}\label{rem:pm2}[$k$-binomial pattern matching]
In their very nice paper \cite{Ftest}, Freydenberger {\em et al.} answer positively to the following questions (similar to Remark~\ref{rem:pm1}):
       \begin{itemize}
         \item Given $k\ge 1$ and two words $u,v$ of length $n$,
           decide, in polynomial time with respect to $n$ and $k$,
           whether or not $u\sim_{k-\mathsf{bin}}v$.
       \item Given $k\ge 1$ and two words $w,x$, find, in polynomial
         time, all occurrences of factors of $w$ which are
         $k$-binomially equivalent to $x$.
          \item Given two $u,v$ of length $n$, find the largest $k$
            such that $u\sim_{k-\mathsf{bin}}v$.
    \end{itemize}
One answer is given by building a non deterministic automaton accepting a language with multiplicities (one counts the number of accepting paths for a given input word) associated with a word $w$ and an integer $k$. This automaton accepts exactly the subwords of $w$ of length at most $k$ and the number of accepting paths of a subword $x$ is precisely $\binom{w}{x}$. The number of states of this automaton is proportional to $|w|\cdot k$. There exist polynomial time procedures to test the equivalence of two such automata \cite{PP1,PP2,PP3} that were initially considered by Sch\"utzenberger for the minimization of weighted automata. Another clever answer is a randomized one based on the evaluation of a polynomial, similar to \eqref{eq:poly}, over a sufficiently large finite field equivalent to the $k$-spectrum (considering evaluation avoids the problem of considering the polynomial as an exponentially growing list of coefficients). Ideas are similar to those found in primality testing algorithms. 
\end{remark}

\begin{definition}
Other related relations exist. The {\em Simon congruence} $\sim_{S}$ is defined as follows. We have $u\sim_S v$ if and only if the series $\sum_{w\in A^*}\binom{u}{w}w$ and $\sum_{w\in A^*}\binom{u}{w}w$ have the same support, i.e., they have the same non-zero binomial coefficients. This congruence has applications to piecewise testable\footnote{A regular language is {\em piecewise testable} if it is a finite Boolean combination of languages of the form $A^*a_1A^*\cdots A^*a_tA^*$ where the $a_i$'s are letters in $A$.} languages. About Q.1.1 and counting the number of classes for $\sim_S$, see \cite{KKS}.
\end{definition}

 \begin{remark}
     Let us mention an extra notion studied by Salomaa in
     \cite{salo2010}. Let $u=u_1\cdots u_n$ be a finite word. The {\em
       sum of the position indices} for the letter $a\in A$ is defined
     by $$S_a(u):=\sum_{i=1}^{|u|} i\, \delta_{u_i,a}.$$
For instance, $S_b(abacbcaba)=2+5+8=15$. Similarly to binomial coefficients, this type of quantity provides information about the positions of occurrences of letters in a word.
 \end{remark}


\section{Parikh matrices}\label{sec:a3}
Let $k\ge 2$ be an integer. Also related to Theorem~\ref{the:parikh} and binomial coefficients, one can extend the abelianization map $\Psi$ as follows. Let $\mathbb{N}^{\ell\times\ell}$ be the monoid of $\ell\times\ell$ matrices equipped with the multiplication of matrices. Let $A_k:=\{a_1,\ldots,a_k\}$ be an ordered finite alphabet. The {\em Parikh matrix mapping} 
$$\psi_k:A^*\to \mathbb{N}^{(k+1)\times(k+1)}$$ is the morphism of monoids defined by the condition: if $\psi_k(a_q) = (m_{i,j})_{1\le i,j\le k+1}$, then for each $i\in\{1,\ldots,k + 1\}$, $$m_{i,i} = 1,\quad m_{q,q+1}=1,$$ all other elements of the matrix $\psi_k(a_q)$ being $0$. There are many papers dealing with Parikh matrices, we only refer to a few of them \cite{SS2,SS3,Salomaa2003,salo,serb}.

\begin{definition}
Two words over $A_k$ are {\em $M$-equivalent}, or {\em matrix equivalent}, if they have the same Parikh matrix. Again, this relation is clearly a congruence because $\psi_k$ is a morphism. If the equivalence class of a word $w$ is reduced to a single element, then $w$ is said to be {\em $M$-unambiguous}.
\end{definition}

  Consider $A=\{a,b\}$ and $a<b$. We have
$$\psi_2(a)=
\begin{pmatrix}
    1&1&0\\0&1&0\\0&0&1\\
\end{pmatrix},\ 
\psi_2(b)=
\begin{pmatrix}
    1&0&0\\0&1&1\\0&0&1\\
\end{pmatrix}\text{ and }\psi_2(abbab)=
\begin{pmatrix}
    1&2&4\\
    0&1&3\\
    0&0&1\\
\end{pmatrix}.$$ The next proposition can be easily deduced from
elementary properties of binomial coefficients of words and matrix
computations. It shows that Parikh matrices for an alphabet of
cardinality $k$ encode $k(k+1)/2$ of the binomial coefficients of a
word $w$ for subwords of length at most $k$. With the above example, the word $abbab$ contains $2$ $a$'s, $3$ $b$'s and $4$ occurrences of the subword $ab$. 
    \begin{theorem}\cite{mat}\label{the:mat}
 Let $w$ be a finite word over $A_k$ and $\psi_k(w) = (m_{i,j})_{1\le i,j\le k+1}$. Then $$m_{i,j+1}=\binom{w}{a_i\cdots a_j}$$ for all $1\le i\le j\le k$.
    \end{theorem}

{\em Generalized Parikh mappings} $\psi_u$, for all words $u\in A^*$ can be defined as follows. Let $u=u_1\cdots u_\ell$. If $\psi_u(a) = (m_{i,j})_{1\le i,j\le \ell+1}$, then for each $i\in\{1,\ldots,\ell + 1\}$, $m_{i,i} = 1$, and for each $i\in\{1,\ldots,\ell\}$, $$m_{i,i+1}=\delta_{a,u_i},$$ all other elements of the matrix $\psi_u(a)$ being $0$.
\begin{remark}
    We get back to the `classical' Parikh matrices over $A_k$ with $u=a_1a_2\cdots a_k$.
\end{remark}
Theorem \ref{the:mat} has the following natural generalization.
    \begin{theorem}\cite{serb}
        Let $u=u_1\cdots u_\ell$ and $w$ a word. Let $\psi_u(w) = (m_{i,j})_{1\le i,j\le \ell+1}$. Then, for all $1\le i\le j\le \ell$, $$m_{i,j+1}=\binom{w}{u_i\cdots u_j}.$$ 

In particular, the first row of $\psi_u(w)$ contains the coefficients corresponding to the prefixes of $u$: $\binom{w}{\varepsilon}$, $\binom{w}{u_1}$, $\binom{w}{u_1u_2}$, \ldots, $\binom{w}{u_1\cdots u_{\ell-1}}$, $\binom{w}u$. Similarly, the last column of  $\psi_u(w)$ contains the coefficients corresponding to the suffixes: $\binom{w}{u}$, $\binom{w}{u_2\cdots u_\ell}$,  \ldots, $\binom{w}{u_1}$, $\binom{w}{\varepsilon}$.
    \end{theorem}
    \begin{example}
Here is an illustration of the latter theorem:
 \begingroup
\renewcommand*{\arraystretch}{1.3}
        $$\psi_{abba}(w)=
    \begin{pmatrix}
        1&\binom{w}{a}&\binom{w}{ab}&\binom{w}{abb}&\binom{w}{abba}\\
        0&1&\binom{w}{b}&\binom{w}{bb}&\binom{w}{bba}\\
        0&0&1&\binom{w}{b}&\binom{w}{ba}\\
        0&0&0&1&\binom{w}{a}\\
        0&0&0&0&1\\
    \end{pmatrix}.$$
\endgroup
    \end{example}
With $\ell$-abelian equivalence and $k$-binomial equivalence, we had two infinite families of refinements. We can also introduce similar refinements, actually uncountably many families of refinements.
    \begin{definition}\label{def:newpsi}
        Let $\mathbf{w}=w_0w_1w_2\cdots $ be an infinite word. Considering the prefixes of $\mathbf{w}$, with this infinite word is associated a sequence of maps
$$(\psi_{w_0\cdots w_j})_{j\ge 0}.$$
We say that two finite words $u,v$ are {\em $(\mathbf{w},j)$-equivalent}, if 
$$\psi_{w_0\cdots w_j}(u)=\psi_{w_0\cdots w_j}(v).$$
This means that $u$ and $v$ have the same binomial coefficients corresponding to the factors occurring in the prefix of length $j+1$ of $\mathbf{w}$.
    \end{definition}
In the above definition, if $\mathbf{w}$ contains every letter of the alphabet, taking $j$ large enough such that every letter of $A_k$ appears in $w_0\cdots w_j$, $(\mathbf{w},j)$-equivalence is a refinement of the abelian equivalence. Note that $u\sim_{\ell-\mathsf{bin}}v$ trivially implies that $u,v$ are $(\mathbf{w},\ell-1)$-equivalent. Also, for every word $\mathbf{w}$, $(\mathbf{w},j+1)$-equivalence is a refinement of $(\mathbf{w},j)$-equivalence (the matrix $\psi_{w_0\cdots w_j}(u)$ is the upper-left corner of $\psi_{w_0\cdots w_jw_{j+1}}(u)$). See Figure~\ref{fig:lattice}.

      \begin{example} Let us illustrate the relations existing between binomial equivalence and $M$-equivalence. Again, these equivalences are, in general, incomparable.
        \begin{itemize}
          \item  The two words $u=abcbabcbabcbab$ and $v=bacabbcabbcbba$ are not $3$-binomially equivalent: $\binom{u}{abb}=34$ and $\binom{v}{abb}=36$ but they share the same Parikh matrix $\psi_3(u)=\psi_3(v)$. This observation only reflects that Parikh matrices encode a fraction of the binomial coefficients. Nevertheless, for a well-chosen generalized Parikh matrix, the two words can of course be distinguished by $\psi_{abb}(u)\neq\psi_{abb}(v)$.

          \item  Erasing the $c$'s in the previous two words, we get two words $u'=abbabbabbab$ and $v'=baabbabbbba$ that are not $3$-binomially equivalent: we again have $\binom{u'}{abb}=34$ and $\binom{v'}{abb}=36$. But they have the same Parikh matrix, i.e., $\psi_2(u)=\psi_2(v)$ (and from the next proposition, we also have that the words are $2$-binomially equivalent). Indeed, $3$-binomial equivalence is a strict refinement of the $2$-binomial equivalence.
          \item  Finally, the two words $u=bccaa$ and $v=cacab$ are not $2$-binomially equivalent: $\binom{u}{ca}=4$ and $\binom{v}{ca}=3$, but they share the same Parikh matrix $\psi_3(u)=\psi_3(v)$.
        \end{itemize}
Also, $\ell$-abelian equivalence and $M$-equivalence are incomparable. Take the same two words as in Example~\ref{exa:kbin}: $abbaba$ and $ababba$ are
$2$-abelian equivalent, but they are not $M$-equivalent.
\end{example}
Over a $2$-letter alphabet, the situation is usually simpler. 

 \begin{prop}
         Over a $2$-letter alphabet, two words are $2$-binomially equivalent if and only if they have the same Parikh matrix, i.e., are $M$-equivalent.
    \end{prop}

    \begin{proof}
        One direction is obvious. Let the alphabet be $\{a,b\}$. Assume that $\psi_2(u)=\psi_2(v)$. We have 
$$\binom{u}{aa}=\binom{|u|_a}{2}=\binom{|v|_a}{2}=\binom{v}{aa}.$$
The same holds for the subword $bb$. We only have to check that $\binom{u}{ba}=\binom{v}{ba}$. This follows from the fact that, for all words $w$, 
$$\sum_{x\in A^2}\binom{w}{x}=\binom{|w|}{2}.$$
    \end{proof}

In particular, the next result completely characterizes the equivalence classes for $2$-binomial equivalence over a $2$-letter alphabet.
\begin{theorem}\cite{salo}\label{pro:salo}
     Over a $2$-letter alphabet $A$, two words are $M$-equivalent if and only if one can be obtain from the other by a finite sequence of transformations of the form $xabybaz\to xbayabz$ where $a,b\in A$ and $x,y,z\in A^*$.
\end{theorem}
As a consequence of this result, a word over a $2$-letter alphabet is $M$-unambiguous (there is no other word with the same Parikh matrix) if and only if it belongs to $a^*b^*+b^*a^*+a^*ba^*+b^*ab^*+a^*bab^*+b^*aba^*$.
\section{Other relations}\label{sec:8}

To conclude with this survey, let us mention a few other relations that can be encountered in combinatorics on words and formal language theory. Berstel and Boasson initiated the study of {\em partial words} containing a `do not know' symbol $\diamond$ serving as a wild card \cite{berboa}. Two words such as 
$$\begin{array}{lccccc}
&a &\diamond& b& b& \diamond\\ \text{and }&\diamond& a& b& \diamond& \diamond\\
\end{array}$$ are {\em compatible} because one can replace the symbols $\diamond$ in such a way that both words match the word $aabba$ (or $aabbb$). This relation `{\em being compatible}' is reflexive and symmetric (but clearly not transitive). Also see, \cite{BS1',BS2,BS3,fw3,fw2}.  
One can generalize this to a {\em similarity relation} associated with a binary relation over an alphabet, see the survey chapter by Halava, Harju and K\"arki in \cite[Chap.~6]{CANT2}. 
\begin{definition}
Consider a reflexive and symmetric relation $R$ over an alphabet $A$. Two words $u_1\cdots u_n$ and $v_1\cdots v_n$ are {\em $R$-similar}, if $(u_i,v_i)\in R$ for all $i$. We write $u\sim_R v$. Partial words corresponds to the special case where the relation $R$ is defined by $(a,\diamond)\in R$, for all $a\in A$.
\end{definition}

\begin{example}
    Assume that $R=\{(a,b),(b,c),(c,d),(d,a)\}$ and take its symmetric and reflexive closure. As an example, we have the following relations:
$$abcd\sim_R bcbd,\ bcbd\sim_R ccac,\ abcd\not\sim_R ccac.$$
\end{example}

One can also think about relations derived from languages or automata. We just give an example. Let $L$ be a language over $A$. The {\em syntactic congruence} is defined as follows. The {\em context} of a word $u$ is the set of pairs of words $(x,y)$ such that $xuy$ belongs to $L$. Two words are syntactically congruent if they have the same context. For instance, see \cite{BerstelReutenauer}.


\begin{thebibliography}{99}


  \bibitem{Aberkane} A. Aberkane, J.D. Currie, N. Rampersad, The
    number of ternary words avoiding abelian cubes grows
    exponentially, {\em J. Integer Seq.} {\bf 7.2} (2004), Art.
    04.2.7.

  \bibitem{Adam2003} B. Adamczewski, Balances for fixed points of
    primitive substitutions, {\em Theoret. Comput. Sci.} {\bf 307}
    (2003), 47--75.

  \bibitem{Allouche93} J.-P. Allouche, Sur la complexit\'e des suites
    infinies, {\em Bull. Belg.  Math. Soc.} {\bf 1} (1994), 133--143.

  \bibitem{AlloucheShallit2} J.-P. Allouche, J. Shallit, Sums of
    digits, overlaps, and palindromes, {\em Discrete Math. Theor.
      Comput. Sci.} {\bf 4} (2000), 1--10.

  \bibitem{AlloucheShallit} J.-P. Allouche, J. Shallit, {\em Automatic
      Sequences: Theory, Applications, Generalizations}, Cambridge
    Univ. Press (2003).

  \bibitem{Arnoux} P. Arnoux, G. Rauzy, Repr\'esentation
    g\'eom\'etrique de suites de complexité $2n+1$, {\em Bull. Soc.
      Math. France} {\bf 119} (1991), 199--215.

  \bibitem{BerstelReutenauer} J. Berstel, C. Reutenauer, {\em
      Noncommutative rational series with applications}, Encycl. of
    Math. and its Appl. {\bf 137}, Cambridge Univ. Press (2011).

  \bibitem{berboa} J. Berstel, L. Boasson, Partial words and a theorem
    of Fine and Wilf, {\em Theoret. Comput. Sci.} {\bf 218} (1999),
    135--141.

  \bibitem{SS1} J. Berstel, J. Karhum\"aki, Combinatorics on words ---
    a tutorial, {\em Bull. Eur. Assoc. Theor.  Comput. Sci.  EATCS}
    {\bf 79} (2003), 178--228.

  \bibitem{berstelV} J. Berstel, L. Vuillon, Coding rotations on
    intervals, {\em Theoret. Comput. Sci.} {\bf 281} (2002), 99--107.

  \bibitem{CANT} V. Berth\'e, M. Rigo (Eds.), {\em Combinatorics, Automata
      and Number Theory}, Encycl. of Math. and its Appl. {\bf 135},
    Cambridge Univ. Press (2010).

  \bibitem{CANT2} V. Berth\'e, M. Rigo (Eds.), {\em Combinatorics, Words and
      Symbolic Dynamics}, Encycl. of Math. and its Appl. {\bf 159},
    Cambridge Univ. Press (2016).

  \bibitem{BS1'} F. Blanchet-Sadri, Codes, orderings and partial words,
    {\em Theoret. Comput. Sci.} {\bf 239} (2004), 177--202.

  \bibitem{BS2} F. Blanchet-Sadri, Periodicity on partial words, {\em
      Comput. Math. Appl.} {\bf 47} (2004), 71--82.

  \bibitem{BS1} F. Blanchet-Sadri, J. D. Currie, N. Rampersad, N.
    Fox, Abelian complexity of fixed point of morphism $0\mapsto
    012$, $1\mapsto 02$, $2\mapsto 1$, Integers {\bf 14} (2014), paper A11.

  \bibitem{BS3} F. Blanchet-Sadri, R. A. Hegstrom, Partial words and a
    theorem of Fine and Wilf revisited, {\em Theoret. Comput. Sci.}
    {\bf 270} (2002), 401--419.

  \bibitem{fw3} F. Blanchet-Sadri, T. Oey, T. Rankin, Fine and Wilf's
    theorem for partial words with arbitrarily many weak periods, {\em
      Internat. J. Found. Comput. Sci.} {\bf 21} (2010), 705--722.

  \bibitem{fw2} F. Blanchet-Sadri, S. Simmons, A. Tebbe, A.
    Veprauskas, Abelian periods, partial words, and an extension of a
    theorem of Fine and Wilf, {\em RAIRO Theor. Inform. Appl.} {\bf
      47} (2013), 215--234.

  \bibitem{Brandenburg} F.-J. Brandenburg, Uniformly growing $k$-th
    power-free homomorphisms, {\em Theoret. Comput. Sci.} {\bf 23}
    (1983), 69--82.

  \bibitem{Brown} T. C. Brown, Is there a sequence on four symbols in
    which no two adjacent segments are permutations of one another?
    {\em Amer. Math. Monthly} {\bf 78} (1971), 886--888.

  \bibitem{Carpi} A. Carpi, On Abelian Power-Free Morphisms, {\em Int.
      J. Algebra Comput.} {\bf 3} (1993), 151--167.

  \bibitem{Carpi2} A. Carpi, On the number of Abelian square-free words on four letters, {\em Discrete Appl. Math.} {\bf 81} (1998), 155--167.

  \bibitem{Cassaigne} J. Cassaigne, Counting overlap-free binary
    words, {\em Lect. Notes in Comp. Sci.} {\bf 665} (1993), 216--225.

  \bibitem{CasSha} J. Cassaigne, J.D. Currie, L. Schaeffer, J.
    Shallit, Avoiding three consecutive blocks of the same size and
    same sum, {\em J. ACM} {\bf 61} (2014), Art. 10.

  \bibitem{CaRa} J. Cassaigne, G. Richomme, K. Saari, L. Q. Zamboni, Avoiding Abelian powers in binary words with bounded Abelian complexity, {\em Int. J. Found. Comp. Sci.} {\bf 22} (2011), 905--920.

  \bibitem{CaSa} J. Cassaigne, J. Karhum\"aki, A. Saarela, On growth
    and fluctuation of $k$-abelian complexity, {\em Lect. Notes in
      Comput. Sci.} {\bf 9139} (2015), 109--122.

  \bibitem{Cobham} A. Cobham, Uniform tag sequences, {\em Math.
      Systems Theory} {\bf 6} (1972), 164--192.

 \bibitem{fw5} S. Constantinescu, L. Ilie, Fine and Wilf's theorem
    for abelian periods, {\em Bull. Eur. Assoc. Theor. Comput. Sci.
      EATCS} {\bf 89} (2006), 167--170.

  \bibitem{CovenHedlund} E. M. Coven and G. A. Hedlund, Sequences with
    minimal block growth, {\em Math. Systems Theory} {\bf 7} (1973),
    138--153.

  \bibitem{CR2} J. D. Currie, N. Rampersad, Fixed points avoiding
    Abelian $k$-powers, {\em J. Combin. Theory Ser. A} {\bf 119}
    (2012), 942--948.

  \bibitem{CRam} J. Currie, N. Rampersad, Growth rate of binary words
    avoiding $xxx^R$, {\em Theoret. Comput. Sci.} {\bf 609} (2016),
    456--468.

  \bibitem{Dekking} F. M. Dekking, Strongly nonrepetitive sequences
    and progression-free sets, {\em J. Combin. Theory Ser. A} {\bf 27}
    (1979), 181--185.

  \bibitem{didier} G. Didier, Combinatoire des codages de rotations,
    {\em Acta Arith.} {\bf 85} (1998), 157--177.

  \bibitem{Dress} A. W. M. Dress, P. L. Erd\H{o}s, Reconstructing
    words from subwords in linear time, {\em Annals of Combinatorics}
    {\bf 8} (2004), 457--462.

  \bibitem{dudik} M. Dudik, L. J. Schulman, Reconstruction from
    subsequences, {\em J. Combin. Theory, Ser. A} {\bf 103} (2003),
    337--348.

  \bibitem{Durand:1998b} F. Durand, A characterization of substitutive
    sequences using return words,{\em Disc. Math.} {\bf 179} (1998),
    89--101.

  \bibitem{Durand} F. Durand, Decidability of the HD0L ultimate
    periodicity problem, {\em RAIRO - Theoret. Inf. and Appl.}  {\bf
      47} (2013), 201--214.


  \bibitem{patmat} T. Ehlers, F. Manea, R. Mercas, D. Nowotka,
    $k$-abelian pattern matching, {\em Lect. Notes Comput. Sci.} {\bf
      8633} (2014), 178--190.

  \bibitem{ELR} A. Ehrenfeucht, K. P. Lee, G. Rozenberg, Subword
    complexities of various classes of deterministic developmental
    languages without interaction, {\em Theoret. Comput. Sci.} {\bf 1}
    (1975), 59--75.

  \bibitem{Erdos1} P. Erd\H{o}s, Some unsolved problems, {\em Michigan
      Math. J.} {\bf 4} (1957), 291--300.


  \bibitem{Entringer} R.C Entringer, D.E Jackson, J.A Schatz, On
    nonrepetitive sequences, {\em J. Combin. Theory Ser. A} {\bf 16},
    (1974), 159--164.

  \bibitem{Fraenkel94} A. S. Fraenkel, R. J. Simpson, How Many Squares
    Must a Binary Sequence Contain?, {\em The Electronic Journal of
      Combinatorics}, {\bf 2}, (1995).

  \bibitem{Ftest} D. D. Freydenberger, P. Gawrychowski, J. Karhum\"aki, F. Manea, W. Rytter, Testing $k$-binomial equivalence, {\tt arXiv:22600.9051}.

  \bibitem{Gre} F. Greinecker, On the $2$-abelian complexity of the
    Thue-Morse word, {\em Theoret. Comput. Sci.} {\bf 593} (2015),
    88--105.

  \bibitem{codes1} V. Halava, T. Harju, T. K\"arki, Relational codes
    of words, {\em Theoret. Comput. Sci.} {\bf 389} (2007), 237--249.

  \bibitem{fw4} V. Halava, T. Harju, T. K\"arki, The theorem of fine
    and Wilf for relational periods, {\em Theor. Inform. Appl.} {\bf
      43} (2009), 209--220.
 
  \bibitem{KarkiR} V. Halava, T. Harju, T. K\"arki, M. Rigo, On the
    periodicity of morphic words, {\em Lect. Notes in Comput. Sci.}
    {\bf 6224} (2010), 209--217.

  \bibitem{HarjuLinna} T. Harju, M. Linna, On the periodicity of
    morphisms on free monoids, {\em RAIRO Inform. Th\'eor. Appl.} {\bf
      20} (1986), 47--54.

  \bibitem{HoltonZamboni} C. Holton, L. Q. Zamboni, Descendants of
    primitive substitutions, {\em Theory Comput. Systems} {\bf 32}
    (1999), 133--157.

  \bibitem{Huo1} M. Huova, Existence of an infinite ternary
    $64$-abelian square-free word, {\em RAIRO - Theoretical
      Informatics and Applications} {\bf 48} (2014), 307--314.

  \bibitem{Huo4} M. Huova, A. Saarela, Strongly $k$-abelian
    repetitions, {\em Lect. Notes in Comput. Sci.} {\bf 8079},
    Springer, (2013).

  \bibitem{Huo2} M. Huova, J. Karhum\"aki, Observations and problems
    on $k$-abelian avoidability, In Combinatorial and Algorithmic
    Aspects of Sequence Processing (Dagstuhl Seminar 11081), (2011)
    2215–-2219.

  \bibitem{Huo3} M. Huova, J. Karhum\"aki, On unavoidability of
    $k$-abelian squares in pure morphic words, {\em J. Integer Seq.}
    {\bf 16} (2013), no. 2, Art. 13.2.9.

  \bibitem{kala}  L. I. Kalashnik, The reconstruction of a word from fragments, in ``Numerical Mathematics and Computer Technology,'' pp. 56--57, Akad. Nauk Ukrain. SSR Inst. Mat., Preprint IV, (1973).

  \bibitem{KKS} P. Karandikar, M. Kufleitner, Ph. Schnoebelen. On the
    index of Simon's congruence for piecewise testability, {\em
      Inform. Processing Let.} {\bf 15} (2015), 515--519.

  \bibitem{Karhu80} J. Karhum{\"a}ki, Generalized Parikh mappings and
    homomorphisms, {\em Inform. and Control} {\bf 47} (1980),
    155--165.

  \bibitem{fw1} J. Karhum\"aki, S. Puzynina, A. Saarela, Fine and
    Wilf's theorem for $k$-abelian periods, {\em Internat. J. Found.
      Comput. Sci.} {\bf 24} (2013), 1135--1152.

  \bibitem{KSZ} J. Karhum{\"a}ki, A. Saarela, L. Q. Zamboni, On a
    generalization of Abelian equivalence and complexity of infinite
    words, {\em J. Combin. Theory Ser. A} {\bf 120} (2013),
    2189--2206.

  \bibitem{KSZ2} J. Karhum{\"a}ki, A. Saarela, L. Q. Zamboni,
    Variations of the Morse-Hedlund theorem for $k$-abelian
    equivalence, {\em Lect. Notes in Comput. Sci.} {\bf 8633} (2014),
    203--214.

  \bibitem{KS} J. Karhum{\"a}ki, J. Shallit, Polynomial versus
    Exponential Growth in Repetition-Free Binary Words, {\em J.
      Combin. Theory Ser. A} {\bf 105} (2004), 335--347.
 
  \bibitem{codes2} T. K\"arki, Compatibility relations on codes and
    free monoids, {\em Theor. Inform. Appl.} {\bf 42} (2008),
    539--552.

  \bibitem{Keranen} V. Ker\"anen, Abelian squares are avoidable on 4
    letters, {\em Lecture Notes in Comput. Sci.}  {\bf 623} (1992),
    41--52.

  \bibitem{PP1} S. Kiefer, A. S. Murawski, J. Ouaknine, B. Wachter, J.
    Worrell, On the complexity of the equivalence problem for
    probabilistic automata, {\em Lect. Notes in Comput. Sci.} {\bf
      7213} (2012), 467--481.

  \bibitem{Kob} Y. Kobayashi, Enumeration of irreducible binary words,
    {\em Disc. Appl. Math.} {\bf 20} (1988), 221--232.

  \bibitem{krasikov} I. Krasikov, Y. Roditty, On a Reconstruction
    Problem for Sequences, {\em J. Combin. Theory, Ser. A} {\bf 77}
    (1997), 344-348.

  \bibitem{Lot} M. Lothaire, {\em Combinatorics on Words}, Cambridge
    Mathematical Library, Cambridge Univ. Press (1997).

  \bibitem{Lot2} M. Lothaire, {\em Algebraic Combinatorics on Words},
    Encycl. of Math. and its Applic. {\bf 90}, Cambridge Univ.
    Press (2002).

  \bibitem{RampersadPF} B. Madill, N. Rampersad, The abelian
    complexity of the paperfolding word, {\em Discrete Math.} {\bf
      313} (2013), 831--838.


  \bibitem{SS2} J. Ma\v{n}uch, Characterization of a word by its
    subwords, in: G. Rozenberg, W. Thomas (Eds.), Developments in
    Language Theory, World Scientific Publ. Co., Singapore, 2000, pp.
    210--219.

    \bibitem{mat} A. Mateescu, A. Salomaa, K. Salomaa, Yu Sheng, A
    Sharpening of the Parikh Mapping, {\em RAIRO-Theoretical Informatics
    and Applications} {\bf 35} (2001), 551--564.

\bibitem{SS3} A. Mateescu, A. Salomaa, S. Yu, Subword histories and
  Parikh matrices, {\em J. Comput. Systems Sci.} {\bf 68} (2004),
  1--21.

\bibitem{MH38} M. Morse,G. A. Hedlund, Symbolic Dynamics, {\em Amer.
    J. Math.} {\bf 60} (1938), 815--866.

  \bibitem{Ochem} P. Ochem, N. Rampersad, J. Shallit, Avoiding
    approximate squares, {\em Internat. J. Found. Comput. Sci.} {\bf
      19} (2008), 633--648.

  \bibitem{Pan1} J.-J. Pansiot, Bornes inf\'erieures sur la
    complexit\'e des facteurs des mots infinis engendr\'es par
    morphismes it\'er\'es, {\em Lect. Notes in Comput. Sci.} {\bf 166}
    (1984), 230--240.

  \bibitem{Pan2} J.-J. Pansiot, Complexit\'e des facteurs des mots
    infinis engendr\'es par morphismes it\'er\'es, {\em Lect. Notes in
      Comput. Sci.} {\bf 172} (1984), 380--389.

  \bibitem{Pansiot86} J.-J. Pansiot, Decidability of periodicity for
    infinite words, {\em RAIRO Inform. Th\'eor. Appl.} {\bf 20}
    (1986), 43--46.

  \bibitem{Parreau} A. Parreau, M. Rigo, E. Rowland, \'E. Vandomme, A
    new approach to the $2$-regularity of the $\ell$-abelian
    complexity of $2$-automatic sequences, {\em Electron. J. Combin.}
    {\bf 22} (2015), paper 1.27.

  \bibitem{Parikh} R. Parikh, On Context-Free Languages, {\em J. of
      the ACM} {\bf 13} (1966).

  \bibitem{PuZ} S. Puzynina, L. Q. Zamboni, Abelian returns in
    Sturmian words, {\em J. Combin. Theory Ser. A} {\bf 120} (2013),
    390--408.

  \bibitem{RRS} N. Rampersad, M. Rigo, P. Salimov, A note on abelian
    returns in rotation words, {\em Theoret. Comput. Sci.} {\bf 528}
    (2014), 101--107.

  \bibitem{Rao15} M. Rao, On some generalizations of abelian power
    avoidability, {\em Theoret. Comput. Sci.} {\bf 601} (2015),
    39--46.

  \bibitem{RaoRS} M. Rao, M. Rigo, P. Salimov, Avoiding 2-binomial
    squares and cubes, {\em Theoret. Comput. Sci.} {\bf 572} (2015),
    83--91.

  \bibitem{Ric1} G. Richomme, K. Saari, L.Q. Zamboni, Balance and
    abelian complexity of the Tribonacci word, {\em Adv. in Appl.
      Math.} {\bf 45} (2010), 212--231.

  \bibitem{Ric2} G. Richomme, K. Saari, L.Q. Zamboni, Abelian
    complexity of minimal subshifts, {\em J. Lond. Math. Soc.} {\bf
      83} (2011), 79--95.

  \bibitem{Rigo} M. Rigo, {\em Formal Languages, Automata and
      Numeration Systems: Introduction to Combinatorics on Words},
    ISTE-Wiley (2014).

  \bibitem{RS} M. Rigo, P. Salimov, Another generalization of abelian
    equivalence: binomial complexity of infinite words, {\em Theoret.
      Comput. Sci.} {\bf 601} (2015), 47--57.

  \bibitem{RSV} M. Rigo, P. Salimov, E. Vandomme, Some properties of
    abelian return words, {\em J. Integer Seq.} {\bf 16} (2013),
    Art. 13.2.5.

  \bibitem{Salomaa2003} A. Salomaa, Counting (scattered) subwords,
    {\em Bull. Eur. Assoc. Theor.  Comput. Sci.  EATCS} {\bf 81}
    (2003), 165--179.

  \bibitem{Salomaa2005} A. Salomaa, Connections between subwords and
    certain matrix mappings, {\em Theoret. Comput. Sci.} {\bf 340}
    (2005), 188--203.

  \bibitem{salo} A. Salomaa, Criteria for the matrix equivalence of
    words, {\em Theoret. Comput. Sci.} {\bf 411} (2010), 1818--1827.

  \bibitem{salo2010} A. Salomaa, Subword balance, position indices and
    power sums, {\em J. Comput. Systems Sci.} {\bf 76} (2010),
    861--871.

  \bibitem{PP2} M.-P. Sch\"{u}tzenberger, On the definition of a
    family of automata, {\em Inf. and Control} (1961), 245--270.

  \bibitem{serb} T.-F. \c{S}erb\u{a}nu\c{t}\u{a}, Extending Parikh
    matrices, {\em Theoret. Comput. Sci.} {\bf 310} (2004), 23--246.

  \bibitem{Sudkamp} T. A. Sudkamp, {\em Languages and Machines: An
      Introduction to the Theory of Computer Science}, Addison-Wesley
    Pub. (1997).

  \bibitem{Thue1} A. Thue, {\"U}ber unendliche Zeichenreihen, {\em
      Norske vid. Selsk. Skr. Mat. Nat. Kl.} {\bf 7} (1906), 1--22.

  \bibitem{Thue2} A. Thue, {\"U}ber die gegenseitige Lage gleicher
    Teile gewisser Zeichenreihen, {\em Norske vid. Selsk.  Skr. Mat.
      Nat. Kl.} {\bf 1} (1912), 1--67.

  \bibitem{Tur1} O. Turek, Abelian complexity function of the Tribonacci word, {\em J. Integer Seq.} {\bf 18} (2015), Art. 15.3.4.

  \bibitem{Tur2} O. Turek, Abelian complexity and abelian co-decomposition, {\em Theoret. Comput. Sci.} {\bf 469} (2013), 77--91.

  \bibitem{PP3}  W. Tzeng, A polynomial-time algorithm for the equivalence of probabilistic automata, {\em SIAM J. Comput.} {\bf 21} (1992), 216--227.

  \bibitem{Vuillon} L. Vuillon, A characterization of Sturmian words
    by return words, {\em European J. Combin.} {\bf 22} (2001),
    263--275. 

\end{thebibliography}
\end{document}